\documentclass[12pt]{article}
\usepackage{amsmath,amssymb,amsthm,bm,cite,fullpage}

\title{Physical Wigner functions}

\author{Carlos L. Benavides-Riveros$^{1,2}$
and Jos\'e M. Gracia-Bond\'ia$^2$
\\ \\
$^1$Zentrum f\"ur Interdisziplin\"are Forschung, Wellenberg 1
\\
Bielefeld 33615, Germany
\\ \\
$^2$Departamento de F\'isica Te\'orica and BIFI research center,
\\
Universidad de Zaragoza, 50009 Zaragoza, Spain
}

\date{\today}

\DeclareMathOperator{\tr}{tr}       

\newcommand{\al}{\alpha}            
\newcommand{\bt}{\beta}             
\newcommand{\dl}{\delta}            
\newcommand{\Ga}{\Gamma}            
\newcommand{\ga}{\gamma}            
\newcommand{\la}{\lambda}           
\newcommand{\om}{\omega}            
\newcommand{\sg}{\sigma}            
\newcommand{\vs}{\varsigma}         

\newcommand{\dn}{{\mathord{\downarrow}}} 
\newcommand{\half}{\tfrac{1}{2}}    
\newcommand{\ihalf}{\tfrac{i}{2}}   
\newcommand{\ox}{\otimes}           
\newcommand{\R}{\mathbb{R}}         
\newcommand{\shalf}{{\scriptstyle\frac{1}{2}}} 
\newcommand{\sihalf}{{\scriptstyle\frac{i}{2}}} 

\newcommand{\up}{{\mathord{\uparrow}}} 
\newcommand{\word}[1]{\quad\mbox{#1}\quad} 
\newcommand{\wh}{\widehat}          
\newcommand{\x}{\times}             
\newcommand{\7}{\dagger}            
\renewcommand{\.}{\cdot}            

\newcommand{\vecform}{\bm}              
\newcommand{\ff}{\vecform{f}}           
\newcommand{\PP}{\vecform{P}}           
\newcommand{\pp}{\vecform{p}}           
\newcommand{\rr}{\vecform{r}}           
\newcommand{\RR}{\vecform{R}}           
\renewcommand{\ss}{\vecform{s}}         
\newcommand{\vv}{\vecform{v}}           
\newcommand{\xx}{\vecform{x}}           
\newcommand{\ZZ}{\vecform{Z}}           
\newcommand{\zz}{\vecform{z}}           

\newtheorem{thm}{Theorem}               

\theoremstyle{definition}
\newtheorem{exa}{Example}               

\makeatletter
\def\section{\@startsection{section}{1}{\z@}{-3.5ex plus -1ex minus
 -.2ex}{2.3ex plus .2ex}{\large\bfseries}}
\def\subsection{\@startsection{subsection}{2}{\z@}{-3.25ex plus -1ex
 minus -.2ex}{1.5ex plus .2ex}{\normalsize\bfseries}}
\makeatother

\begin{document}

\maketitle

\begin{abstract}
In spite of their potential usefulness, the characterizations of
Wigner functions for Bose and Fermi statistics given by O'Connell and
Wigner himself almost thirty years ago~\cite{Freudianslip} has drawn
little attention. With an eye towards applications in quantum
chemistry, we revisit and reformulate them in a more convenient way.
\end{abstract}

\section{Introduction and history}
\label{sec:introibo}

By definition, a spin-zero $n$-body Wigner quasiprobability
distribution (or Wigner function for short) is given in terms of the
density matrix in configuration or momentum space, respectively
$\rho,\,\wh\rho$, by~\cite{Wigner32}:
\begin{align}
W_\rho(\xx;\pp) := \frac{1}{\pi^{dn}} \int \rho(\xx - \zz;\xx + \zz)\,
e^{2i\pp\.\zz} \,d\zz = \frac{1}{\pi^{dn}} \int \wh\rho\,(\pp -
\zz;\pp + \zz)\, e^{-2i\xx\.\zz} \,d\zz,
\label{eq:onsen} 
\end{align}
with the notation $\xx = (\xx_1,\dots,\xx_n)$ for $n$ bodies, where
$\xx_i\in\R^d$ (say $d = 3$ for ordinary space), and similarly for
$\pp$ and~$\zz$. For a pure state one has $\rho(\xx;\xx')=\Psi(\xx)
\Psi^*(\xx')$, with~$\Psi$ the corresponding wave function. We have
taken units so that $\hbar=1$. The relation $\rho\leftrightarrow
W_\rho$ is one-to-one, being the restriction to the convex set of
positive operators of unit trace of a linear isomorphism of functions
of two sets of variables, essentially the inverse of the unitary
\textit{Wigner transformation}.%
\footnote{V\'arilly and one of us proved it in \cite{Phoebe} to be of
order 24.}

Averages of Wigner functions with classical phase space observables
reproduce the expected values predicted by standard quantum mechanics.
This is why they have become an important tool, successfully adopted
in statistical physics, quantum optics \cite{corazondeleon} and now
chemistry \cite{Pluto,Hermione}. It is however not easy to
characterize them, although necessary and sufficient conditions for a
phase space function to be an admissible Wigner quasiprobability
distribution are known \cite{NarcoOC86, Titania,Pluto}.

Now, a natural question is: when does a Wigner function 
correspond to a wave function symmetric or antisymmetric under
permutations of its variables? This was posed since the early days, in
view of applications: see the references in~\cite{Freudianslip}.
But only the latter article purported to offer a general answer.%
\footnote{We leave aside the second-quantized approach to Wigner
quasiprobability, which has known scant success.}

On their characterization, O'Connell and Wigner wrote: ``It must be
admitted that this equation for the distribution function, postulating
the Bose statistics for a system of spin~0 particles, is much more
complicated that the corresponding equation for the density matrix''.
After discussing one-half spin systems, towards the conclusion of the
paper, twice they repeat this gloomy assessment almost verbatim. It is
true that their key formulas~(14a) or~(14b), together with their~(11),
look rather unwieldy. This moreover seems to have discouraged
borrowing of the second and more interesting part of their paper, on
systems of spin~$1/2$ particles, hardly exploited elsewhere. They
offered no examples.

\smallskip

Of late, the development of density functional theory based on 1-body
Wigner functions led us to reexamine the matter. Their
$N$-representability conditions are well understood.%
\footnote{Including non-ensemble aspects~\cite{Klyachko}.}
Existence and some of the properties of the Wigner function energy
functional were established in~\cite{Pluto}. The theory has the
flavour of an almost exact Thomas--Fermi formalism in phase space,
needing ``only'' to incorporate electron correlation.

We give here a \textit{simple} answer to the question of quantum
statistics for Wigner functions. We will be dealing mainly with
identical fermions, for which the Wigner function is a spin multiplet;
hence definition~\eqref{eq:onsen} will be insufficient. Even so, our
characterization takes the form of mere preservations or changes of
sign under permutation of two variables ---just as in the ordinary
formalism of quantum mechanics. This makes it trivial that the square
of such a permutation induces the identity, which is not at all obvious
in~\cite{Freudianslip}. So we throw long-due light on the achievements
of that paper, hoping to rescue from near-oblivion its insights.

The summary of the article is as follows. We deal first with the
conditions for symmetric or antisymmetric scalar Wigner functions
---both are required for quantum chemistry purposes. We illustrate our
contentions with a few example classes of concrete Wigner functions in
Section~3. Then we go on to Wigner spin orbitals, revisiting the
second part of~\cite{Freudianslip}. In Section~5 we exemplify again.
Section~6 is the conclusion.

\section{The basic theorems}
\label{sec:loshuesos}

It will be enough to consider the 2-body problem. Bringing in mean and
difference coordinates, or, in chemists' jargon, \textit{extracule}
and \textit{intracule} coordinates, respectively given~by
\begin{equation}
\RR = \frac{1}{\sqrt2}(\xx_1 + \xx_2),  \qquad
\rr = \frac{1}{\sqrt2}(\xx_1 - \xx_2),
\label{eq:culeable} 
\end{equation}
with the customary abuse of notation, the symmetry\slash antisymmetry
conditions (say, on configuration space) for spinless bodies
respectively read
$\rho(\RR,\rr;\RR',\rr') = \pm\rho(\RR,-\rr;\RR',\rr')$ or
$\rho(\RR,\rr;\RR',\rr') = \pm\rho(\RR,\rr;\RR',-\rr')$. Together they
imply
\begin{equation}
\rho(\RR,\rr;\RR',\rr') = \rho(\RR,-\rr;\RR',-\rr');
\label{eq:antisic} 
\end{equation}
and reciprocally, the latter indistinguishability property together with
either of the above conditions implies the other.

It is not hard to see that with
\begin{equation}
\PP = \frac{1}{\sqrt 2}(\pp_1 + \pp_2),  \qquad
\pp = \frac{1}{\sqrt 2}(\pp_1 - \pp_2),
\label{eq:onemoment} 
\end{equation}
the meaning of $W(\RR,\rr;\PP,\pp)$ is unambiguous. This is due to the
linear symplectic invariance of the Wigner function formalism. Then
\eqref{eq:antisic} is equivalent to
\begin{equation}
W(\RR,\rr;\PP,\pp) = W(\RR,-\rr;\PP,-\pp).
\label{eq:kiboshi} 
\end{equation}
Since the discussion turns around the intracule variables, it is worth
regarding $\RR,\PP$ as parameters, introducing the following notation:
$$
\om_{\RR,\PP}(\rr,\pp) := W(\RR,\rr;\PP,\pp).
$$
Let us invoke the following partial Fourier transform on the intracule
set of variables:
$$
\tilde\om_{\RR,\PP} (\vv,\pp) := \int \om_{\RR,\PP} (\rr,\pp)
\,e^{2i\vv\.\rr} \,d\rr = \tilde\om_{\RR,\PP} (-\vv,-\pp).
$$
The last equality is seen to hold when \eqref{eq:antisic} or
equivalently \eqref{eq:kiboshi} hold, and reciprocally. Now we have
two momentum-like intracular variables, and the following appears
natural.

\begin{thm} 
A scalar Wigner 2-body function comes from a density matrix
symmetric\slash antisymmetric in its first set of variables,
respectively in its
second set, if and only if, for all $\vv$ and $\pp$:
\begin{equation}
\tilde\om_{\RR,\PP}(\vv,\pp) = \pm \tilde\om_{\RR,\PP}(\pp,\vv);
\word{respectively} \tilde\om_{\RR,\PP}(\vv,\pp) =
\pm \tilde\om_{\RR,\PP}(-\pp,-\vv).
\label{eq:manga} 
\end{equation}
\end{thm}

\begin{proof}
Consider the following integral:
\begin{align*}
\tilde\om_{\RR,\PP}(\vv,\pp) &= \frac1{\pi^{2d}} \int \rho\big(\RR
- \ZZ, \rr - \zz;\RR + \ZZ, \rr + \zz\big)\, e^{2i\PP\.\ZZ +
2i\pp\.\zz} \,e^{2i\vv\.\rr} \,d\ZZ \,d\zz \,d\rr
\\
&= \frac{\pm1}{\pi^{2d}} \int \rho\big(\RR - \ZZ, \zz - \rr;\RR + \ZZ,
\zz + \rr\big)\, e^{2i\PP\.\ZZ + 2i \vv\.\rr}
\,e^{2i\pp\.\zz}\,d\ZZ \,d\rr \,d\zz
\\
&= \pm \int \om_{\RR,\PP}(\zz,\vv) \,e^{2i\pp\.\zz}\,d\zz =:
\pm  \tilde\om_{\RR,\PP} (\pp,\vv).
\end{align*}
Thus necessity of the first condition is proved. Conversely, given
that $\rho\leftrightarrow W_\rho$ is one-to-one, it is readily seen
that~\eqref{eq:manga} holds only if $\rho$ is respectively 
symmetric\slash antisymmetric. The proof of the second condition is
similar. Clearly, if we assume
$\tilde\om_{\RR,\PP}(\vv \pp)=\tilde\om_{\RR,\PP}(-\vv,-\pp)$, either
of the conditions of~\eqref{eq:manga} implies the other. Needless to
say, one may formulate the conclusion analogously in terms of
$\hat\om_{\RR,\PP}(\rr,\ss):=\int\om_{\RR,\PP}(\rr;\pp)\,
e^{-2i\ss\.\pp}\,d\pp$.
\end{proof}

\smallskip

Our result extends to $n$-body functions by just considering intracule
and extracule coordinates for the first pair of adjacently labeled
particles. That is, we require only one condition of the type
\eqref{eq:manga}, together with the indistinguishability
condition~\eqref{eq:kiboshi} for all intracules.

\section{Examples}
\label{sec:lacarne}

Use of Gaussian basis sets in density functional theory with Wigner
functions is if anything more natural than in standard quantum
chemistry~\cite{Jensen}. This motivates our first example.

\begin{exa} 
Take as a boson-type wave function the symmetric product of two
general Gaussians centered at the origin:
\begin{equation}
\Psi(x_1,x_2)
= C \bigl( \psi_1(x_1)\psi_2(x_2) + \psi_1(x_2)\psi_2(x_1) \bigr),
\label{eq:gauss-pair} 
\end{equation}
where, for $j = 1,2$:
$$
\psi_j(x) = \frac{d^{1/4}_j}{\pi^{1/4}}\,
e^{-\shalf d_j x^2 - \sihalf b_j d_j x^2}
\word{with} d_j > 0, \ b_j \in \R.
$$
(The normalization factor $C$ is unimportant here.) The corresponding
2-body quasidensity~is
\begin{align}
W(x_1,x_2; p_1,p_2)
&\propto W_{11}(x_1;p_1) W_{22}(x_2;p_2)
+ W_{22}(x_1;p_1) W_{11}(x_2;p_2)
\notag \\
&\qquad + W_{12}(x_1;p_1) W_{21}(x_2;p_2)
+ W_{21}(x_1;p_1) W_{12}(x_2;p_2).
\label{eq:boson} 
\end{align}
Here $W_{jk}$ represents an \textit{interference}, namely,
\begin{gather*}
W_{jk}(x,p) = \frac{d^{1/4}_j d^{1/4}_k}{\pi\,d^{1/2}_{jk}}\, 
e^{-A_{jk} x^2 - 2B_{jk} xp - d_{jk}^{-1} p^2},  \word{where}
\\
d_{jk} := \half(d_j + d_k) + \ihalf(b_j d_j - b_k d_k),  \qquad
b_{jk} := \half(b_j d_j + b_k d_k) - \ihalf(d_j - d_k),
\\
A_{jk} := d_{jk} + b_{jk}^2/d_{jk},  \qquad  B_{jk} := b_{jk}/d_{jk}.
\end{gather*}
The quadratic form in the exponent of the $W_{jk}$ is given by a
symmetric, symplectic matrix with positive definite real
part~\cite{Pluto}. When $k =j$, we have a Gaussian \textit{pure
state},
$$
W_{jj}(x,p) = \pi^{-1}\,e^{-(d_j + b^2_jd_j) x^2 - 2b_j xp - d_j^{-1}
p^2};
$$
whose coefficient matrix is real, symplectic and positive definite
\cite{Robert}.

To see that the quasidensity \eqref{eq:boson} fulfils
\eqref{eq:manga}, change variables according to \eqref{eq:culeable}
and~\eqref{eq:onemoment}, and let $\la_{ijkl}(R,r;P,p) :=
W_{ij}(x_1;p_1) W_{kl}(x_2;p_2)$. Now, multiplying by $e^{2ivr}$ and
integrating with respect to~$r$, we obtain, after a little work,
$$
\int \la_{jjkk}(R,r;P,p) \, e^{2ivr} \,dr = 
\int \la_{jkkj}(R,r;P,v) \, e^{2ipr} \,dr \word{with} k \neq j,
$$
thereby verifying condition~\eqref{eq:manga} for this example.
\textit{Mutatis mutandis}, Gaussian sets like the ones in
\eqref{eq:gauss-pair} with a minus instead of a plus sign exemplify
the antisymmetric case.
\end{exa}

\begin{exa} 
In the early years of Quantum Mechanics, as a prolegomenon to
calculating the energy levels for helium, Heisenberg~\cite{Heisenberg}
studied the \textit{harmonium}, an exactly integrable analogue of a
two-electron atom. It exhibits two fermions interacting with an
external harmonic potential and repelling each other by a Hooke-type
force. Being simple, but not trivial, this system has been borrowed in
many contexts. It is sometimes called the ``Moshinsky atom'', since
Moshinsky reintroduced it with the purpose of studying correlation
energy~\cite{Moshinsky68,Hermione}. Also, it has been recruited to
investigate Bose--Einstein condensation~\cite{RS00}, black-hole
entropy~\cite{Srednicki93} and sundry issues in quantum chemistry
---see~\cite{Laetitia,Marmulla} and references therein.

The harmonium Hamiltonian in Hartree-like units is given by:
\begin{equation*}
H(\xx_1,\xx_2;\pp_1,\pp_2) = \frac{|\pp_1|^2}{2} + \frac{|\pp_2|^2}{2}
+ \frac{k}{2}\big(|\xx_1|^2 + |\xx_2|^2\big) - \frac{\dl}{4}\, |\xx_1
- \xx_2|^2.
\end{equation*}
Introducing extracule and intracule coordinates and the frequencies
$\nu := \sqrt k$ and $\mu := \sqrt{k - \dl}$, the Hamiltonian is
rewritten as that of two independent oscillators:
$$
H = H_R + H_r := \frac{|\PP|^2}{2} + \frac{\nu^2|\RR|^2}{2} +
\frac{|\pp|^2}{2} + \frac{\mu^2|\rr|^2}{2}\,.
$$
Since the problem factorizes completely, we work in dimension one. The
\textit{orbital} part of such an eigenfunction is written
$\phi_n(R)\psi_m(r)$, with the parity of~$\psi_m(r)$ even for spin
singlet states and odd for triplet states. Wigner quasiprobabilities
associated to those eigenvectors have the general form:
$W_n(R,P)W_m(r,p)$, where, with $L_n$ denoting the $n$-th Laguerre
polynomial:
\begin{align*}
W_n(R,P) = \frac{(-1)^n}{\pi} \, L_n(4H_R/\nu)\, e^{-2H_R/\nu}, \qquad
W_m(r,p) = \frac{(-1)^m}{\pi} \, L_m(4H_r/\mu)\, e^{-2H_r/\mu}.
\end{align*}
Defining
$$
\Ga_m(v,p) = (-1)^m \int W_m(r,p)\, e^{2ivr} \,dr
= \frac{1}{\pi} \int L_m(4H_r/\mu)\, e^{-2H_r/\mu}\, e^{2ivr} \,dr,
$$
by use of the generating function of the Laguerre polynomials we
obtain
\begin{align*}
\sum_{m=0}^\infty \Ga_m(v,p)\, x^m &= \frac{1}{\pi(1 - x)} \int
e^{-4(H_r/\mu)x/(1-x)} e^{-2(H_r/\mu)} e^{2ivr} \,dr
\\
&= \frac{1}{\pi(1 - x)}\, e^{-\frac{1+x}{1-x}\,p^2/\mu} \int
e^{-\frac{1+x}{1-x}\,\mu r^2 + 2ivr} \,dr
\\
&= \frac{1}{\sqrt{\pi\mu(1 - x^2)}}\, \exp\biggl( -
\frac{1+x}{1-x}\,\frac{p^2}{\mu} - \frac{1-x}{1+x}\,\frac{v^2}{\mu}
\biggr) = \sum_{m=0}^\infty (-)^m \Ga_m(p,v)\, x^m.
\end{align*}
Thus $\Ga_m(v,p)=-\Ga_m(p,v)$ for $m$ odd and $\Ga_m(v,p)= \Ga_m(p,v)$
for $m$ even, and whenever the wave function $\Psi_{nm}\equiv\phi_n
\psi_m$ is symmetric / antisymmetric, the corresponding Wigner
functions $W_{nm}(R,P;r,p)=W_n(R,P)W_m(r,p)$ in agreement
with~\eqref{eq:manga} do respectively satisfy:
$$
\int W_{nm}(R,P;r,p)\,e^{2ivr}\,dr = \pm\int W_{nm}(R,P;r,v)\,e^{2ipr}
\,dr.
$$
\end{exa}

\section{Spin Wigner functions}
\label{sec:lamadredelcordero}

The standard definition for \textit{spin} Wigner functions, found for
instance in the seminal work on atomic Wigner functions~\cite{SD87},
regards the latter (just as the density matrices) as $2^n\x 2^n$ matrices
in spin space. 
\begin{align*}
\mathcal{W}^{\vs_1,\dots,\vs_n;\vs'_1,\dots,\vs'_n}_\rho(\xx_1,\dots,\xx_n;
\pp_1, \dots,\pp_n) := \frac1{\pi^{dn}} \int\rho\big(\xx - \zz;\vs_1,
\dots,\vs_n; \xx + \zz,\vs'_1,\dots,\vs'_n\big)e^{2i\pp.\zz} \,d\zz.
\end{align*}
Here $\vs$ and $\vs'$ denote the discrete spin variables. In
particular, a $1$-body atomic Wigner distribution in matrix form would
be of the form
$$
\mathcal{W}^{(1)} = \begin{pmatrix} W^{\up_1\up_{1'}}(\xx;\pp) &
W^{\up_1\dn_{1'}}(\xx,\pp) \\
W^{\dn_1\up_{1'}} (\xx,\pp) &
W^{\dn_1\dn_{1'}}(\xx,\pp) \end{pmatrix};
$$
and a $2$-body atomic Wigner distribution:
\begin{align*}
\mathcal{W}^{(2)} = \begin{pmatrix} W^{\up_1\up_2\up'_1\up'_2}(1,2) &
W^{\up_1\up_2\up'_1\dn'_2}(1,2) & W^{\up_1\up_2\dn'_1\up'_2}(1,2) &
W^{\up_1\up_2\dn'_1\dn'_2}(1,2) \\
W^{\up_1\dn_2\up'_1\up'_2}(1,2) & W^{\up_1\dn_2\up'_1\dn'_2}(1,2) & 
W^{\up_1\dn_2\dn'_1\up'_2}(1,2) & W^{\up_1\dn_2\dn'_1\dn'_2}(1,2) \\
W^{\dn_1\up_2\up'_1\up'_2}(1,2) & W^{\dn_1\up_2\up'_1\dn'_2}(1,2) & 
W^{\dn_1\up_2\dn'_1\up'_2}(1,2) & W^{\dn_1\up_2\dn'_1\dn'_2}(1,2) \\
W^{\dn_1\dn_2\up'_1\up'_2}(1,2) & W^{\dn_1\dn_2\up'_1\dn'_2}(1,2) & 
W^{\dn_1\dn_2\dn'_1\up'_2}(1,2) & W^{\dn_1\up_2\dn'_1\dn'_2}(1,2) 
\end{pmatrix},
\end{align*}
where $(1,2)$ on the right hand side standing for the orbital phase
space variables. We normalize them by
$\tr\int\mathcal{W}{(1)} \,d\xx \,d\pp = 1$,
$\tr\int\mathcal{W}{(2)} \,d1 \,d2 = 1$ (not quite the custom in
chemistry). Symmetry of~$\rho$ under interchange of \textit{both}
orbital and spin variables entails:
\begin{equation}
\mathcal{W}^{(2)}_{\rho}(1,2) = A \,\mathcal{W}^{(2)}_{\rho}(2,1)\, A,
\word{where} A = \begin{pmatrix} 1 & 0 & 0 & 0 \\ 0 & 0 & 1 & 0 \\
0 & 1 & 0 & 0 \\ 0 & 0 & 0 & 1\end{pmatrix}.
\label{eq:banzai} 
\end{equation}

The matrix approach contains some redundancies in practice, and was
implicitly criticized by Wigner in his last
years~\cite{Freudianslip,yellowed}. He sought instead to endow the
spin Wigner functions with ostensible physical meaning, by arranging
their entries into tensors under the rotation group. Given the
essentially unitary matrix,
$$
U := \frac{1}{2} \bigl(\sg^\kappa_{\vs\vs'}\bigr)
= \frac{1}{2} \begin{pmatrix} 1 & 0 & 0 & 1 \\
0 & 1 & 1 & 0 \\ 0 & i & -i & 0 \\ 1 & 0 & 0 & -1 \end{pmatrix}
\word{where $\kappa=0,x,y,z$, \, and $UU^\7=1/2$,}
$$
for the 1-body quasiprobability these are provided by
$$
\begin{pmatrix} W^0 \\ W^x \\ W^y \\ W^z \end{pmatrix} = U
\begin{pmatrix} W^{\up_1\up_{1'}} \\
W^{\up_1\dn_{1'}} \\ W^{\dn_1\up_{1'}} \\
W^{\dn_1\dn_{1'}} \end{pmatrix}.
$$
There the entries on the right hand side are not real in general; but
on the left side they are. Matters turn interesting for the 2-body
function, whereupon
\begin{equation}
\begin{pmatrix} W^0 \\ W^x \\ W^y \\ W^z \end{pmatrix} \ox
\begin{pmatrix} W^0 \\ W^x \\ W^y \\ W^z \end{pmatrix} = \big(U \ox
U\big) \begin{bmatrix}\begin{pmatrix} W^{\up_1\up_{1'}} \\
W^{\up_1\dn_{1'}} \\ W^{\dn_1\up_{1'}} \\ W^{\dn_1\dn_{1'}}
\end{pmatrix} \ox \begin{pmatrix} W^{\up_2\up_{2'}} \\
W^{\up_2\dn_{2'}} \\ W^{\dn_2\up_{2'}} \\ W^{\dn_2\dn_{2'}}
\end{pmatrix}\end{bmatrix}.
\label{eq:sayonara} 
\end{equation}
The central question addressed by O'Connell and Wigner is the
transformation of $U\ox U$ under particle exchange $1\leftrightarrow
2$; this is better answered in terms of the physical tensor components
of~$\mathcal W$. Denoting representations of the rotation group by
their dimension, and since the 1-body function is the sum of one
rotational scalar and one vector part, the addition rule for angular
momentum yields:
$$
\bigl( [{\bf1}] \oplus [{\bf3}] \bigr)^{\otimes2}
= 2[{\bf1}] \oplus 3[{\bf3}] \oplus [\bf5];
$$
that is two scalars, three vectors and one quadrupole (symmetric
traceless tensor). Let now $W^{00}$ replace $W^0\ox W^0$ in the
notation, and so on. We reorganize the left hand side
of~\eqref{eq:sayonara} as a spin multiplet:
\begin{align*}
W^{sc1} &= W^{00} - W^{xx} - W^{yy} - W^{zz},
\\
W^{sc2} &= \tfrac13\big(3W^{00} + W^{xx} + W^{yy} + W^{zz}\big),
\\
W^{v1}  &= \big(W^{x0} + W^{0x}, W^{y0} + W^{0y}, 
W^{z0} + W^{0z}\big) 
\\
W^{v2}  &= \big(W^{x0}_- + iW^{zy}_-, W^{y0}_- + iW^{xz}_-,
W^{0z}_- + iW^{xy}_-\big)
\\
W^{v3}  &= \big(W^{x0}_- - iW^{zy}_-, W^{y0}_- - iW^{xz}_-,
W^{z0}_- - iW^{yx}_-\big)
\\
W^q &= \big(-W^{xx} - W^{yy} + 2W^{zz}, W^{xy} + W^{yx}, W^{yz} +
W^{zy}, W^{xx} - W^{yy}, W^{xz} + W^{zx}\big),
\end{align*}
with $W^{x0}_-:=W^{x0}-W^{0x},\,W^{xy}_-:=W^{xy}-W^{yx}$, and so on.
The first two terms of the multiplet are the scalars, then the three
vectors, and the quadrupole in a standard presentation.

In summary, collecting $(sc1,sc2,v1,v2,v3,q)\equiv\ff$, for us an
electronic 2-body Wigner function is a multiplet denoted
$\om(\rr,\pp;\ff)$, the extracule labels being suppressed. The Fermi
symmetry condition for the exchange of one set of spin coordinates
$\vs_1\leftrightarrow \vs_2$ \textit{and} of the spatial coordinates,
borrowing the notation used in the spin-zero case, reads:
$$
\tilde\om(\vv,\pp;\ff) :=
-\tilde\om\bigl(\pp,\vv;\ff_{\vs_1\leftrightarrow\vs_2}\bigr).
$$
Then the exchange transformation rule for the Wigner function
multiplet comes out even simpler, in that there are fewer minus
signs than the one for the density matrix:
\begin{align}
\begin{pmatrix} \tilde\om^{sc1}(\vv,\pp) \\
\tilde\om^{sc2}(\vv,\pp) \\ \tilde\om^{v1}(\vv,\pp) \\
\tilde\om^{v2}(\vv,\pp) \\ \tilde\om^{v3}(\vv,\pp) \\
\tilde\om^q(\vv,\pp) \end{pmatrix} = \begin{pmatrix} +1 &&&&& \\ & -1
&&&& \\ && -1 &&& \\ &&& -1 && \\ &&&& +1 & \\ &&&&& -1 \end{pmatrix}
\begin{pmatrix} \tilde\om^{sc1}(\pp,\vv) \\
\tilde\om^{sc2}(\pp,\vv) \\ \tilde\om^{v1}(\pp,\vv) \\ 
\tilde\om^{v2}(\pp,\vv) \\ \tilde\om^{v3}(\pp,\vv) \\ 
\tilde\om^q(\pp,\vv) \end{pmatrix}.
\label{eq:gaijin} 
\end{align}
This because $\tilde\om^{sc1}$ is odd under $\vs_1\leftrightarrow
\vs_2$, while $\tilde\om^{sc2}$ is even, and so on. Of course, one can
choose to impose the Fermi condition on the \textit{primed} spin
coordinates. Then $\tilde\om^{v2},\,\tilde\om^{v3}$ are peculiar in
that they become respectively odd and even. But the general
indistinguishability condition~\eqref{eq:banzai} now implies
\begin{align*}
\begin{pmatrix} \tilde\om^{sc1}(\vv,\pp) \\
\tilde\om^{sc2}(\vv,\pp) \\ \tilde\om^{v1}(\vv,\pp) \\
\tilde\om^{v2}(\vv,\pp) \\ \tilde\om^{v3}(\vv,\pp) \\
\tilde\om^q(\vv,\pp) \end{pmatrix} = \begin{pmatrix} +1 &&&&& \\ & 1
&&&& \\ && 1 &&& \\ &&& -1 && \\ &&&& -1 & \\ &&&&& 1 \end{pmatrix}
\begin{pmatrix} \tilde\om^{sc1}(-\vv,-\pp) \\
\tilde\om^{sc2}(-\vv,-\pp) \\ \tilde\om^{v1}(-\vv,-\pp) \\ 
\tilde\om^{v2}(-\vv,-\pp) \\ \tilde\om^{v3}(-\vv,-\pp) \\ 
\tilde\om^q(-\vv,-\pp) \end{pmatrix};
\end{align*}
and this saves the day.

\section{Examples}
\label{sec:lacarne2}

Formula~\eqref{eq:gaijin} is well adapted to the needs of quantum
chemistry since the standard Hamiltonian there does not contain spin
coordinates; thus one uses a spin-restricted formalism~\cite{pirita},
with the same set of symmetric or antisymmetric spatial orbitals for
``up'' and ``down'' spins. Then several components of the Wigner
multiplet vanish.

For a two-fermion system, the singlet pure spin state is of the form
\begin{align*}
\mathcal W^\mathrm{singlet} = \half\big(\up_1\up_{1'}\dn_2\dn_{2'} -
\up_1\dn_{1'}\dn_2\up_{2'} - \dn_1\up_{1'}\up_2\dn_{2'} +
\dn_1\dn_{1'}\up_2\up_{2'}\big)W = W^{sc1}.
\end{align*}
The only non-zero contribution is given by the first scalar, and in
the occasion the Pauli principle naturally reads $\tilde\om^{\rm
singlet}(\vv,\pp)=\tilde\om^{\rm singlet}(\pp,\vv)$. The one-body
Wigner distribution for this state is just~$W^0$.

For triplet states, one deals with a linear superposition of symmetric
spin states, namely,
$$
|\Psi\rangle = \al |\up_1 \up_2\rangle + \bt \frac{|\up_1 \dn_2\rangle +
|\dn_1 \up_2\rangle}{\sqrt 2} + \ga|\dn_1 \dn_2\rangle
\word{provided that} |\al|^2  + |\bt|^2 + |\ga|^2 = 1.    
$$
In terms of the Wigner spin multiplet, it reads:
\begin{align*}
\mathcal W^\mathrm{triplet} = &\big(|\al|^2 + |\ga|^2\big)
\big(W^{sc2} + \tfrac13 W^{q,1}\big) + |\bt|^2 \big(W^{sc2} - \tfrac23
W^{q,1}\big)
\\
+ & \tfrac1{\sqrt2}\big(\al^* \bt + \al \bt^* + \ga^*\bt +
\ga\bt^*\big) \, W^{v1,x} + \tfrac{i}{\sqrt2}\big(\al^* \bt - \al
\bt^* - \ga^*\bt + \ga\bt^*\big) \, W^{v1,y}
\\
+ & \big(|\al|^2 - |\ga|^2\big) \, W^{v1,z} + \tfrac{i}{\sqrt2}
\big(\al^* \bt - \al \bt^* + \ga^*\bt - \ga\bt^*\big) \, W^{q,3}
\\
+ & i\big(\ga\al^* - \ga^*\al\big) \, W^{q,2} + \big(\ga\al^* +
\ga^*\al\big) \, W^{q,4} + \tfrac1{\sqrt2}\big(\al^* \bt + \al \bt^* -
\ga^*\bt - \ga\bt^*\big) \, W^{q,5}.
\end{align*}
In the multiplet expansion of the triplet there appear only the second
scalar, the first vector and the quadrupole. They all carry minus
signs in~\eqref{eq:gaijin}, and the transformation rule reads
$\tilde\om^\mathrm{triplet}(\vv;\pp)=-\tilde\om^\mathrm{triplet}
(\pp;\vv)$. Its one-body distribution is equal to
\begin{align*}
W^0 + \tfrac1{\sqrt2}(\al^* \bt + \al \bt^* + \ga^*\bt +
\ga\bt^*) W^x + \tfrac{i}{\sqrt2}(\al^* \bt - \al \bt^* -
\ga^*\bt + \ga\bt^*)W^y + (|\al|^2 - |\ga|^2)W^z .
\end{align*}
Notice how different are the $\beta =\gamma = 0$ or
$\beta = \alpha = 0$ states, leading respectively to
$W^{00} + W^{zz} \pm (W^{z0} + W^{0z})$,
from the $\alpha = \gamma = 0$ state, to which there corresponds
$W^{00}+W^{xx}+ W^{yy}-W^{zz}$. The fact that these states belong in
different strata under rotations, somewhat hidden in the Hilbert space
formalism ---see the discussion in~\cite[Sect.~7.7.c]{BL}--- is
here~apparent.

\section{Conclusion}
\label{sec:vale-vale}

On the theoretical side, the little attention received by
paper~\cite{Freudianslip} has concerned mostly its adaptation to the
spherical Moyal formalism for spin, developed by V\'arilly and one of
us in~\cite{Miranda}. For a recent example see~\cite{Outbrackpaper},
and~\cite{12104075} for a nice treatment of that formalism emphasizing
its connection to the tensor operators~\cite{Blum}. The middle path
followed here, in the footsteps of Wigner and O'Connell, displays the
physical appeal and information of the spherical method; and it
appears better adapted to the needs of quantum chemistry. We can only
speculate that it essentially coincides with the (so far, unpublished)
approach ``by~the theory of group
representations''~\cite[Ch.~8]{AnnMoyalbook} arrived at by Moyal in
his last years.

\subsection*{Acknowledgments}

We are grateful to Joseph C. V\'arilly for a careful reading of
the manuscript. CLBR~has been supported by a Banco Santander
scholarship. JMGB has been supported by a grant from the regional
government of Arag\'on. He owes as well to the Zentrum f\"ur
interdisziplin\"are Forschung, for support and warm hospitality.

\end{document}